\documentclass[11pt]{article}

\usepackage{fullpage}
\usepackage[utf8]{inputenc}

\usepackage{algorithm}
\usepackage{algorithmic}

\usepackage{tikz}
\usetikzlibrary{automata}
\tikzset{
		>=latex,
		blk/.style={state,fill=black},
		wht/.style={state,fill=white},
		node distance=0.5\columnwidth
}

\usepackage{amsmath}
\usepackage{amsfonts}
\usepackage{amssymb}
\usepackage{graphicx}
\usepackage{float}
\usepackage{amsthm}
\usepackage{enumitem}

\newtheorem{Lemma}{Lemma}
\newtheorem{Corollary}{Corollary}
\newtheorem{Theorem}{Theorem}
\newtheorem{Definition}{Definition}

\title{Optimally Gathering Two Robots}
\author{Adam Heriban$^\star$ \and Xavier D\'{e}fago$^\dag$ \and S\'{e}bastien Tixeuil$^\star$}
\date{$^\star$UPMC Sorbonne Universit\'{e}s, France\\
$^\dag$Tokyo Institute of Technology, Japan}

\begin{document}

\maketitle

\begin{abstract}
We present an algorithm that ensures in finite time the gathering of two robots in the non-rigid ASYNC model. To circumvent established impossibility results, we assume robots are equipped with 2-colors lights and are able to measure distances between one another. Aside from its light, a robot has no memory of its past actions, and its protocol is deterministic. Since, in the same model, gathering is impossible when lights have a single color, our solution is optimal with respect to the number of used colors. 
\end{abstract}

\section{Introduction}

Networks of mobile robots evolving in a 2-dimensional Euclidean space recently captured the attention of the distributed computing community, as they promise new applications (rescue, exploration, surveillance) in potentially dangerous (and harmful) environments. Since its initial presentation~\cite{suzuki99}, this computing model has grown in popularity and many refinements have been proposed (see~\cite{FPS12b} for a recent state of the art). From a theoretical point of view, the interest lies in characterizing the exact conditions for solving a particular task. 

In the model we consider, robots are anonymous (\emph{i.e.}, indistinguishable from each-other), oblivious (\emph{i.e.}, no persistent memory of the past is available), and disoriented (\emph{i.e.}, they do not agree on a common coordinate
system). The robots operate in Look-Compute-Move cycles. In each cycle a robot ``Looks'' at its surroundings and obtains (in its own coordinate system) a snapshot containing the locations of all robots. Based on this visual information, the robot ``Computes'' a destination location (still in its own coordinate system) and then ``Moves'' towards the computed
location. Since the robots are identical, they all follow the same deterministic algorithm. The algorithm is oblivious if
the computed destination in each cycle depends only on the snapshot obtained in the current cycle (and not on the past
history of execution). The snapshots obtained by the robots are not consistently oriented in any manner (that is, the robots' local coordinate systems do not share a common direction nor a common chirality%
\footnote{Chirality denotes the ability to distinguish left from right.}%
).

The execution model significantly impacts the ability to solve collaborative tasks. Three different levels of synchronization have been considered. The strongest model~\cite{suzuki99} is the fully synchronized (FSYNC) model where each phase of each cycle is performed simultaneously by all robots. The semi-synchronous (SSYNC) model~\cite{suzuki99} considers that time is discretized into rounds, and that in each round an arbitrary yet non-empty subset of the robots are active. The robots that are active in a particular round perform exactly one atomic Look-Compute-Move cycle in that round. It is assumed that the scheduler (seen as an adversary) is fair in the sense that in each execution, every robot is activated
infinitely often. The weakest model is the asynchronous model~\cite{FPS12b} (ASYNC), which allows arbitrary delays between the Look, Compute and Move phases, and the movement itself may take an arbitrary amount of time. In this paper, we consider the most general ASYNC model. 

\paragraph{Related Work.}

The gathering problem is one of the benchmarking tasks in mobile robot networks, and has received a considerable amount of
attention (see~\cite{FPS12b} and references herein). The gathering tasks consist in all robots (each considered as a dimensionless point in a 2-dimensional Euclidean space) reaching a single point, not known beforehand, in finite time. 
A foundational result~\cite{suzuki99,CRTU15j} shows that in the SSYNC model, no deterministic algorithm can solve gathering for two robots without additional assumptions. This impossibility result naturally extends to the ASYNC model. 

In hostile environments such as those we envision, robots are likely to fail. So far, three kinds of failures were considered in the context of deterministic gathering~\cite{AP06j,BDT13c,DP12j,BCM15c,DGC+15,DGM+06}:
\begin{enumerate}
\item \emph{Transient faults}: robots may experience transient faults that corrupt their current state, leading to an arbitrary initial configuration, from which algorithms may not recover (\emph{e.g.}, if the initial configuration is bivalent, and the execution model is SSYNC). As a result, the only known self-stabilizing solution~\cite{DP12j} (that is, able to recover from \emph{any} initial configuration) in the classical fair model considers only odd sets of robots (so, a bivalent configuration cannot exist in this setting). Other self-stabilizing algorithms impose stricter assumptions on the scheduler \cite{DGC+15,DGM+06}.
\item \emph{Crash faults}: when robots may stop executing their algorithm unexpectedly (and correct robots are not able to
  distinguish a correct robot from a crashed one at first sight), guaranteeing that correct robots still gather in finite time is a challenge. All proposed solutions so far~\cite{AP06j,BDT13c,BCM15c,BT15,DGC+15,DGM+06} are restricted to the SSYNC model.
\item \emph{Byzantine faults}: when robots may have completely arbitrary
(and possibly malicious) behavior, there exists no deterministic gathering protocol in the SSYNC model even assuming that at most one robot may be Byzantine~\cite{AP06j}. This impossibility extends to the ASYNC case.  
\end{enumerate} 

To circumvent the aforementioned impossibility results, it was proposed to endow each robot with a \emph{light}~\cite{DAS2016171}, that is, it is capable of emitting a fixed number of colors visible to all other robots. This additional capacity allows to solve gathering of two robots in the most general ASYNC model, provided that robots lights are capable to emit at least \emph{four} colors. In the more restricted SSYNC model, Viglietta~\cite{Viglietta2014} proved that being able to emit two colors is sufficient to solve the gathering problem. In the same paper~\cite{Viglietta2014}, he also proved that an algorithm that only makes use of observed colors to decide on its next move cannot gather two robots in the ASYNC model using only two colors. Finally, Viglietta proves~\cite{Viglietta2014} that three colors and the ability to detect null distances is sufficient in ASYNC. Both solutions in ASYNC~\cite{DAS2016171,Viglietta2014} and SSYNC~\cite{Viglietta2014} output a correct behavior independently of the initial value of the lights' colors. Recently, Okumura \emph{et al.}~\cite{OWK17r} presented an algorithm with two colors that gathers robots in ASYNC assuming \emph{rigid} moves (that is, the move of every robot is never stopped by the scheduler before completion), or assuming non-rigid moves but robots are aware of $\delta$ (the minimum distance before which the scheduler cannot interrupt their move). Also, the solution of Okumura \emph{et al.}~\cite{OWK17r} requires lights to have a specific color in the initial configuration. 

The remaining open case is the feasibility of gathering with only two colors in the most general ASYNC model, without additional assumptions.

\paragraph{Our Contribution.}

\begin{table}
\newcommand\ASYNC{\textbf{ASYNC}}
\newcommand\SSYNC{\emph{SSYNC}}
\newcommand\YES{\emph{yes}}
\newcommand\No{\textbf{NO}}
\centering\small
\resizebox{\linewidth}{!}{
\begin{tabular}{|c|l|l|l|l|l|}\hline
\textbf{Reference}  & \textbf{Synchrony} & \textbf{Rigid} &
\textbf{Initial Color} & \textbf{$\delta$ known} & \textbf{\# Colors}\\
\hline
\cite{DAS2016171}   & \ASYNC    & \No   & \No   & \No   & \emph{4} \\\hline
\cite{Viglietta2014}& \SSYNC    & \No   & \No   & \No   & \textbf{2 (opt.)} \\\hline
\cite{Viglietta2014}& \ASYNC    & \No   & \No   & \No   & \emph{3} \\\hline
\cite{OWK17r}       & \ASYNC    & \YES  & \YES  & \No   & \textbf{2 (opt.)} \\\hline
\cite{OWK17r}       & \ASYNC    & \No   & \YES  & \YES  & \textbf{2 (opt.)} \\\hline
\textbf{This paper} & \ASYNC    & \No   & \No   & \No   & \textbf{2 (opt.)} \\\hline 
\end{tabular}
}
\caption{Gathering robots with lights}
\label{tab:related}

\end{table}

We present a solution to the gathering of two robots that is optimal with respect to all considered criteria: it uses the minimum number of colors (compared to the $4$ colors required by the algorithm of Das \emph{et al.}~\cite{DAS2016171} and the $3$ colors of the algorithm of Viglietta~\cite{Viglietta2014} that operates in the same system ASYNC setting), it performs in the most general ASYNC model (unlike the SSYNC algorithm of Viglietta~\cite{Viglietta2014}), while still not requiring additional assumptions such as rigid moves, initial colors, or knowledge about $\delta$ (unlike the algorithms of Okumura \emph{et al.}~\cite{OWK17r}). Characteristics of our solution and comparison with previous work are summarized in Table~\ref{tab:related}, where boldface characteristics are optimal.  

\section{Model}

We consider the ASYNC model of robots with lights introduced in~\cite{DAS2016171}. In more details, we assume a system of two robots endowed with colored lights (with two possible colors, Black and White) that are modeled as mobile colored points in the Euclidean plane $\mathbb{R}^2$. 

Both robots execute cycles that consist of three phases : LOOK, COMPUTE and MOVE. When a robot is not executing a LOOK-COMPUTE-MOVE (LCM) cycle, it is considered to be in a WAIT phase. Each phase can be described as follows:
\begin{itemize}
\item WAIT: The robot is idle and waiting for activation.
\item LOOK: The robot is taking a snapshot of the position of the other robot, as well as the colors of both robots. This phase is assumed instantaneous.
\item COMPUTE: The robot computes its next destination using the snapshot. A robot is able to change the color of its own light at the end of its COMPUTE phase. 
\item MOVE: The robot moves towards its destination.
\end{itemize}
The duration of the COMPUTE and MOVE phases, and the delay between two phases, are chosen by an adversary and can be arbitrary long, but finite. The adversary decides when robots are activated assuming a \emph{fair} scheduling \emph{i.e.}, in any configuration, all robots are activated within finite time. The adversary also controls the robots movement along their target path and can stop a robot before reaching its destination, but not before traveling at least a distance $\delta > 0$ ($\delta$ being unknown to the robots).
In other words, if a robot has a target at a distance $x$, we assume that, at the end of its MOVE phase, the robot has moved a distance in the interval $[min(\delta,x),x]$. The exact position reached is determined by the adversary scheduler.

Robots are anonymous, meaning they are indistinguishable and they execute the exact same algorithm. However, for the sake of practicality, they are called $A$ and $B$ in the sequel. Each robot has a local coordinate system about which we make no assumptions, in particular, $A$ and $B$ may have distinct North, chirality, and unit distance. We also assume that, except from their light color, robots have no mean of explicitly communicating with each other.

The robots are also oblivious, in that they do not remember their past actions. This implies that the COMPUTE phase can have no other input than the snapshot from the last LOOK phase.

\section{Our Algorithm}

\subsection{Previous Results}

Viglietta observes~\cite{Viglietta2014} that, in order to solve the gathering problem in SSYNC, an algorithm must accomplish two things: 

\begin{itemize}
\item In case robots are synchronized, they need to move towards the midpoint.
\item In case robots are activated alternatively, one needs to move towards the other. In that case, the other robot must not move.
\end{itemize}

In ASYNC, Viglietta also shows~\cite{Viglietta2014} that no algorithm using only two colors can solve gathering if the destination computation solely relies on this form of calculation:

\[
me.destination = (1-\lambda ) \cdot me.position + \lambda \cdot other.position
\]

With:

\[
\lambda = f(me.color, other.color)
\]

Where $f$ is a function (that is, it associates to a 2-tuple a single image).

It is similarly assumed that the next color of a robot only depends on the current colors of the two robots and not on the distance between the robots.

Algorithms that follow these rules of computation are called class $\mathcal{L}$ algorithms. Then, the only algorithm of class~$\mathcal{L}$ that satisfies these criteria is presented in Figure~\ref{fig:viglietta}.

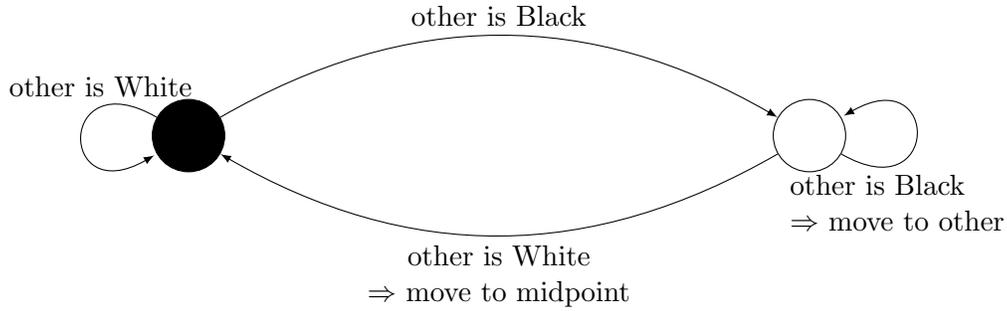
\begin{figure}[htbp]
	\centering
	\begin{tikzpicture}
		\node[blk] (B) {};
		\node[wht] (W) [right of=B] {};	
		\path[->] (B) edge[bend left] node[above,align=center]{other is Black} (W);
		\path[->] (W) edge[bend left] node[below,align=center]{other is White\\$\Rightarrow$ move to midpoint} (B);
		\path[->] (W) edge[out=-30,in=30,loop] node[near start,below,align=left]{other is Black\\$\Rightarrow$ move to other} (W);
		\path[->] (B) edge[out=150,in=210,loop] node[near start,above,align=right]{other is White} (B);
	\end{tikzpicture}
	\caption{Viglietta's~\cite{Viglietta2014} Algorithm}
	\label{fig:viglietta}
\end{figure}

Now, there exists an execution of this algorithm that does not solve ASYNC gathering (see Lemma 4.9 in Viglietta's paper~\cite{Viglietta2014}) when both robots start in the Black color:

\begin{enumerate}
\item Let both robots perform a Look phase, so that both will turn White.
\item Let robot $A$ finish the current cycle and perform a new Look, while the $B$ waits. Hence, $A$ remains White and moves to $B$’s position. Now, we
let $B$ finish the current cycle and perform a new Look. So, $B$ turns Black and moves to the midpoint $m$ between $A$ and $B$. 
\item Let $A$ finish the current cycle, thus reaching $B$, and perform a whole new cycle, thus turning Black. 
\item Finally, let $B$ finish the current cycle, thus turning Black and moving to $m$. 
\end{enumerate}

As a result, both robots are again set to Black, are in a Wait phase, both have executed at least one cycle, and their distance has halved. Thus, by repeating the same pattern of moves, they approach one another but never gather.

Because of this execution, it is not possible to solve gathering with two colors with an $\mathcal{L}$ class algorithm.

As a result, we do not design our algorithm to be of class $\mathcal{L}$, as our computation of the next color not only depends on the respective colors of the two robots, but also on their respective distance. 

\subsection{Our Algorithm}

We observe that in the problematic aforementioned execution, there is an instant when both robots are actually gathered, but are later separated because of pending moves. 

We thus introduce a behavior change in the White state of Viglietta's~\cite{Viglietta2014} algorithm to obtain our proposal, presented in Figure~\ref{fig:ouralgorithm} and Algorithm~\ref{alg1}.

\begin{figure}[htpb]
	\centering
	\begin{tikzpicture}
		\node[blk] (B) {};
		\node[wht] (W) [right of=B] {};	
		\path[->] (B) edge[bend left] node[above,align=center]{other is Black} (W);
		\path[->] (W) edge[bend left] node[below,align=left]{other is White $\wedge$ $\neg$Gathered:\\$\Rightarrow$ move to midpoint} (B);
		\path[->] (W) edge[out=75,in=15,loop] node[near start,above,align=left]{other is Black:\\$\Rightarrow$ move to other} (W);
		\path[->] (W) edge[out=-75,in=-15,loop] node[below,align=left]{Gathered\\$\Rightarrow$ do nothing} (W);
		\path[->] (B) edge[out=150,in=210,loop] node[near start,above,align=right]{other is White} (B);
	\end{tikzpicture}
	\caption{Our Algorithm}
	\label{fig:ouralgorithm}
\end{figure}
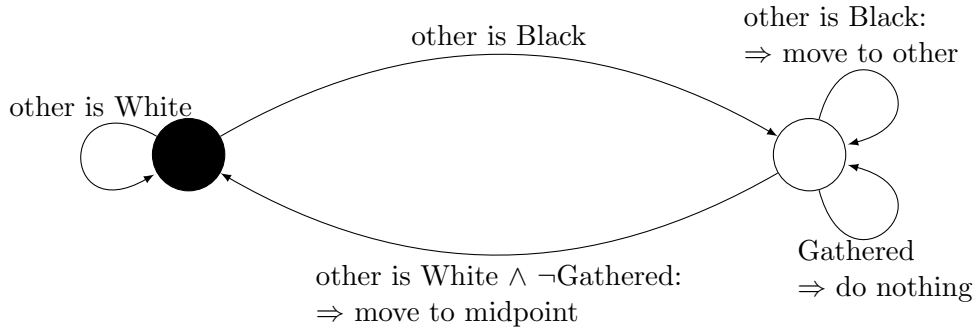

\begin{algorithm}
\caption{ASYNC Robot gathering with two colors}         
\label{alg1}         
\begin{algorithmic}                   
    \STATE 
    \IF{me.color = White}
    	\IF{(me.position = other.position)}
    		\STATE do nothing
		\ELSIF{other.color = White}
        	\STATE me.destination $\Leftarrow$ other.position/2
        	\STATE me.color $\Leftarrow$ Black
    	\ELSIF{other.color = Black}
        	\STATE me.destination $\Leftarrow$ other.position
		\ENDIF
    \ELSIF{me.color = Black \AND other.color = Black}
        \STATE me.color $\Leftarrow$ White
    \ENDIF
\end{algorithmic}
\end{algorithm}

Our proposal breaks the infinite loop in the problematic execution, as it prevents robot $A$ from switching to color Black after reaching $B$ and forces it to remain White. This implies that activating $B$ afterwards actually separates the robots into different colors, and prevents them from going back to both being Black.

Let us observe that our new algorithm no longer belongs to class $\mathcal{L}$, since the same observed $2$-tuple of colors may yield different outcomes depending on the distance between $A$ and $B$. In particular, when both robots are observed White, the next color depends on whether the two robots are gathered. So, the assumption of Viglietta~\cite{Viglietta2014} that a new color is solely determined by the current colors no longer holds.%
\footnote{It is worth noting that, while the definition of class~$\mathcal{L}$ does not explicitly mention that the new color is also obtained as a function of the two observed colors, the Lemma~4.4 of Viglietta's paper \cite{Viglietta2014} entirely relies on this implicit fact, and so does the 3-color algorithm for the ASYNC model.
}

We now need to prove that this new algorithm actually solves the gathering problem in ASYNC in a self-stabilizing manner. Our main result can be stated as follows:

\begin{Theorem}
\label{thm:gathering}
Algorithm~\ref{alg1} solves the gathering problem for two robots in a self-stabilizing fashion for the non-rigid ASYNC model.
\end{Theorem}

\section{Correctness}

\subsection{Describing configurations}

We wish to prove the self-stabilizing property of our algorithm in the most general ASYNC model for the gathering specification. As a result, possible configurations must include all possible relative positions of the two robots, but also their current light color, and most importantly their current phase (and phase status) in the look-compute-move cycle. Indeed, the current phase can be seen as a \emph{program pointer} whose value can be initially corrupted in an arbitrary initial setting. A direct consequence of this observation is that all possible combinations of those parameters must be considered as possible initial configurations. Let us first list in Figure~\ref{fig:table} the various phases that can be reached by the algorithm, for two robots $A$ and $B$:

\begin{figure}[H]
	\centering
	\includegraphics[width=0.7\linewidth]{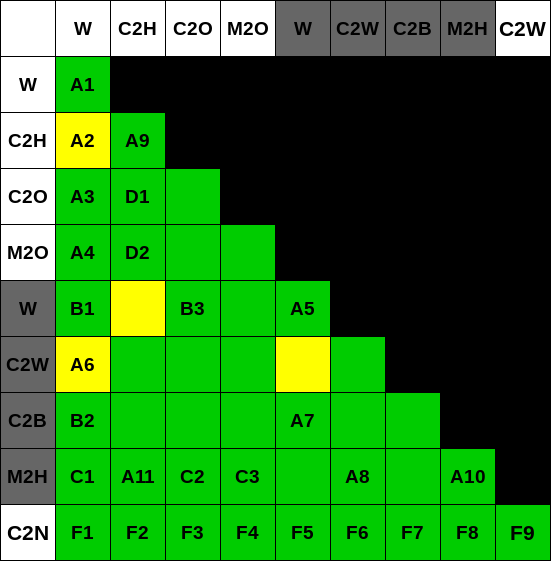}
	\caption{Possible configurations}
	\label{fig:table}
\end{figure}

The possible phases for a robot $A$ are:
\begin{itemize}
\item \textbf{W}: Wait,
\item \textbf{C2H}: Compute to Half (at the end of his COMPUTE phase, $A$ switches its color to Black, and moves towards the midpoint),
\item \textbf{C2O}: Compute to Other (at the end of his COMPUTE phase, $A$ stays White and moves towards the other robot),
\item \textbf{M2O}: Move to the Other robot's position,
\item \textbf{C2W}: Compute to White (COMPUTE phase that leads to color change and no motion),
\item \textbf{C2B}: Compute to Black (COMPUTE phase that leads to no color change and no motion),
\item \textbf{M2H}: Move to Half point,
\item \textbf{C2N}: Compute to Nothing (COMPUTE phase that leads to no motion). 
\end{itemize}

Since robots are anonymous, we only need to consider half of the possible configurations (\emph{i.e.}, the combination $(\mathbf{W},\mathbf{C2H})$ is the same as $(\mathbf{C2H},\mathbf{W})$). 
In a given configuration, the scheduler may activate either $A$, $B$, or both. In most cases, activating both robots has the same effect as activating $A$ then $B$, or $B$ then $A$. However, in a few cases, the outcome of simultaneously activating $A$ and $B$ is not deterministic, and may lead to two different configurations. Those non-deterministic configuration are outlined in yellow color in Figure~\ref{fig:table}. 

To ease the description and reasoning about the various phases of the algorithm, we divide the complete set of configurations into six subsets. 
Each subset of configurations is drawn as a graph whose vertices represent configurations and (directed) edges represent the possibility to reach another configuration by executing the algorithm of either one or both robots (3 cases). Some vertices appear in several subsets as the subsets we consider do not form a partition of all configurations. To facilitate the reading of each subset, we adopted a color code for vertices described in Figure~\ref{fig:colorcode}, and the rest of the diagram notations are presented in Figure~\ref{fig:nomenclature}. 

\begin{figure}[htbp]
	\centering
	\includegraphics[width=0.7\columnwidth]{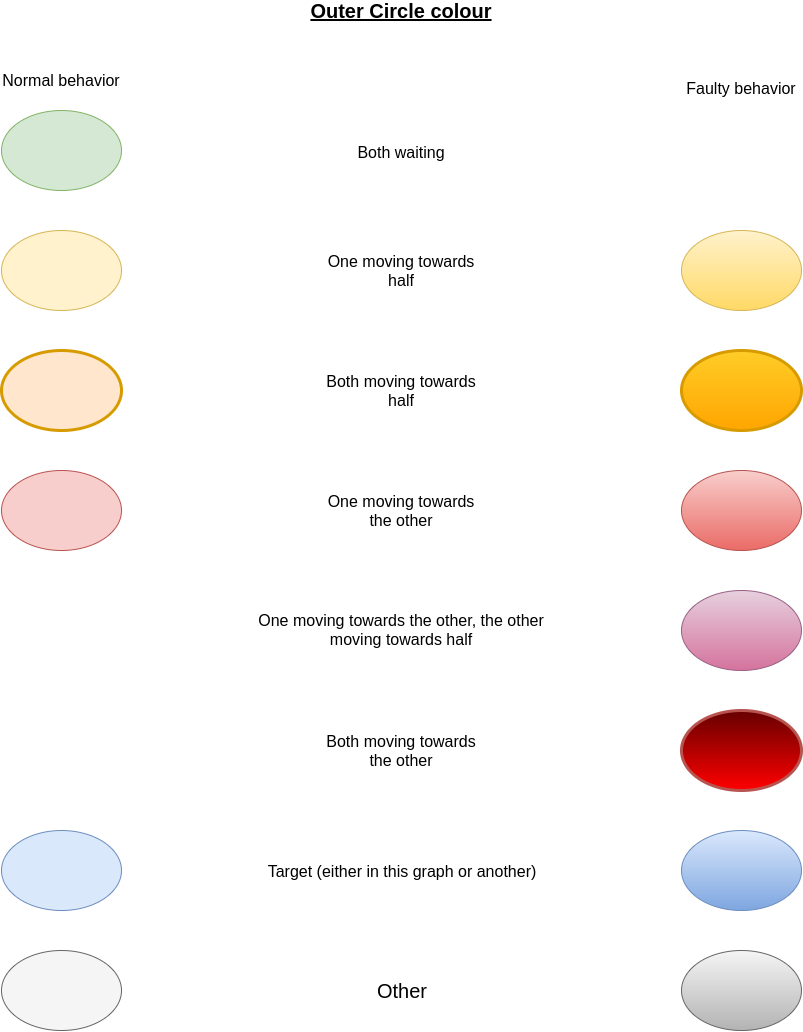}
	\caption{Vertex color code}
	\label{fig:colorcode}
\end{figure}

\begin{figure}[htbp]
	\centering
	\includegraphics[width=0.7\columnwidth]{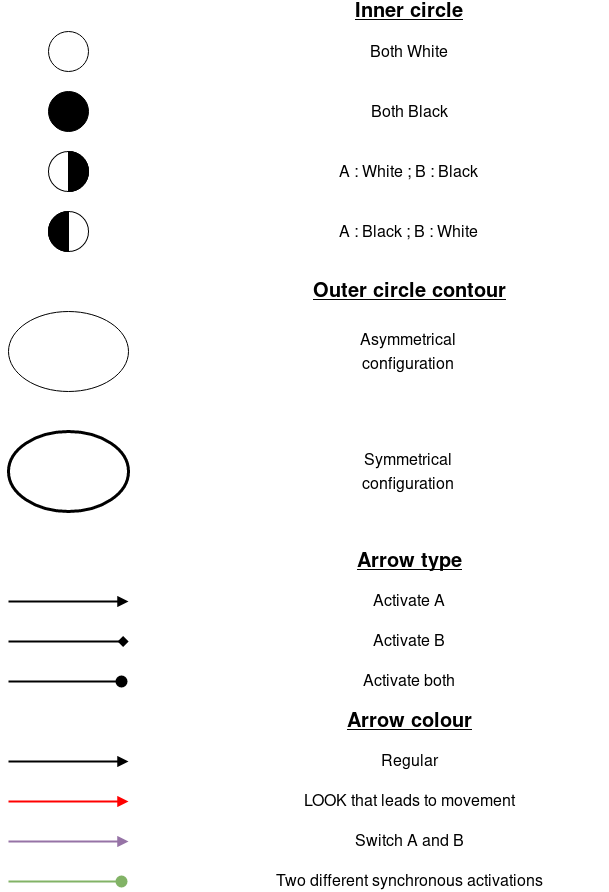}
	\caption{Graph nomenclature}
	\label{fig:nomenclature}
\end{figure}

Each graph is divided in two parts : the top part, called nominal behavior and a bottom part called terminal behavior. The behavior of the algorithm is called \emph{nominal} when the robots are not gathered. Some nodes, represented in the bottom part, change their behavior if the gathering is achieved during their activation. This is the case \emph{e.g.} when a White robot starts its LOOK phase and perceives the other robot and itself on the same location. Specific gathered behavior is called \emph{terminal}. 

\begin{Definition}[Nominal Configuration]
A configuration $c$ is \emph{nominal} when the two robots are not located in the same position in $c$.
\end{Definition}

\begin{Definition}[Terminal Configuration]
A configuration $c$ is \emph{terminal} when the two robots are located in the same position in $c$.
\end{Definition}

\begin{Definition}[Valid Configurations]
A subset ${\mathcal C}_v$ of configurations is \emph{valid} if for any configuration $c\in{\mathcal C}_v$, either \emph{(i)} $c$ is nominal and every execution leads to a terminal configuration in a finite number of steps, or \emph{(ii)} $c$ is terminal and mandates the robots to eventually remain in a gathered configuration forever.
\end{Definition}

\begin{Definition}[Faulty Configurations]
A subset ${\mathcal C}_v$ of configurations is \emph{faulty} if for at least one configuration $c\in{\mathcal C}_v$, $A$ and $B$ are moving towards different targets.
\end{Definition}

We now list the subsets we consider in the sequel, assuming $A$ and $B$ denotes the two robots in the system:

\begin{enumerate}[label=(\alph*)]
\item \textbf{SYM}: configurations where $A$ and $B$ can remain synchronized indefinitely. Those configurations would be considered the regular behavior in a FSYNC model where robots initially share the same color. The \textbf{SYM} configurations are presented in Figure~\ref{fig:SYM}.

\item \textbf{ASYM}: configurations where $A$ and $B$ are different colors, and no possibility of being synchronized again. The \textbf{ASYM} configurations are presented in Figure~\ref{fig:ASYM}.

\item \textbf{FAULTY 1}: configurations that can be reached after a White to Black de-synchronization which allow different targets for $A$ and $B$ and cannot lead back to \textbf{SYM}. The \textbf{FAULTY 1} configurations are presented in Figure~\ref{fig:FAULTY1}.

\item \textbf{FAULTY 2}: configurations that can be reached after a Black to White de-synchronization and allow different targets for $A$ and $B$. The \textbf{FAULTY 2} configurations are presented in Figure~\ref{fig:FAULTY2}.

\item \textbf{ILLEGAL}: configurations that cannot be reached by the algorithm but need to be taken into account as possible starting configurations to ensure self-stabilization. The \textbf{ILLEGAL} configurations are presented in Figure~\ref{fig:ILLEGAL}.

\item \textbf{GATHERED}: configurations that can only be reached if the gathering is complete. The \textbf{GATHERED} configurations are presented in Figure~\ref{fig:GATHERED}.
  
\end{enumerate}

\begin{figure*}[htbp]
	\centering
	\includegraphics[width=0.8\linewidth]{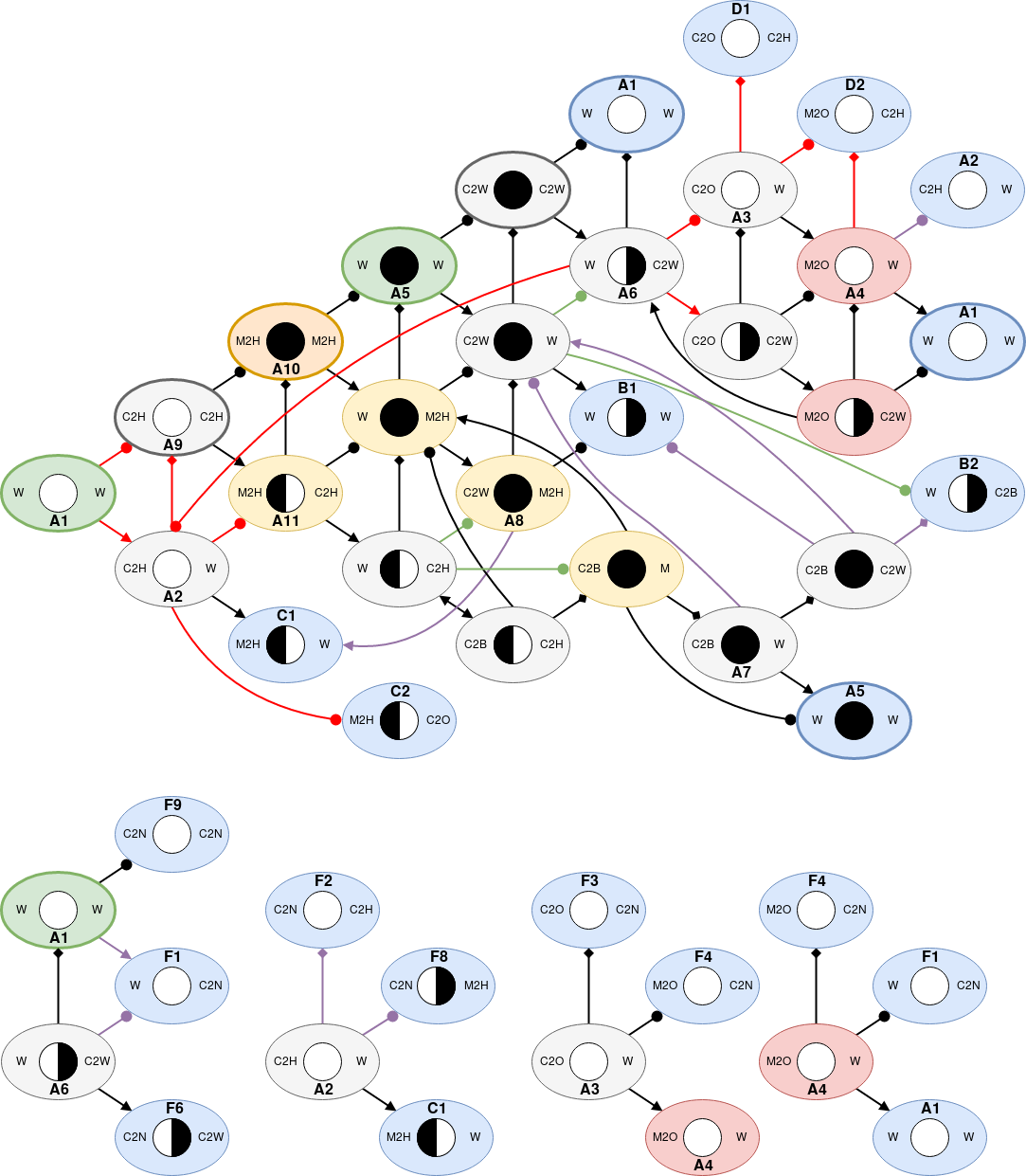}
	\caption{\textbf{SYM} configurations}
	\label{fig:SYM}
\end{figure*}

\begin{figure}[htbp]
	\centering
	\includegraphics[width=0.6\columnwidth]{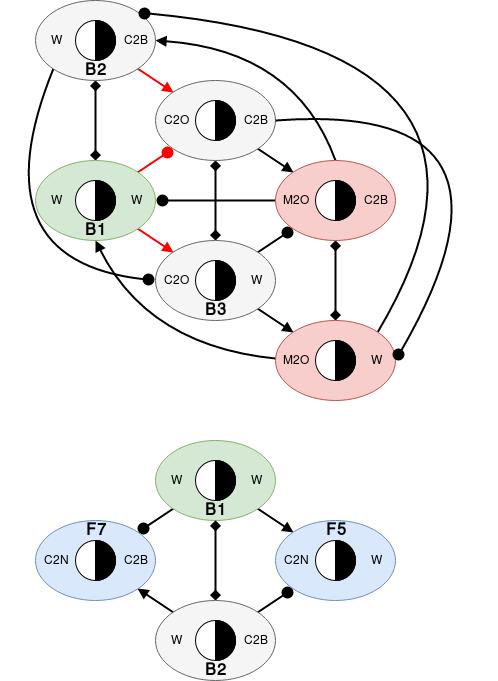}
	\caption{\textbf{ASYM} configurations}
	\label{fig:ASYM}
\end{figure}

\begin{figure}[htbp]
	\centering
	\includegraphics[width=0.6\columnwidth]{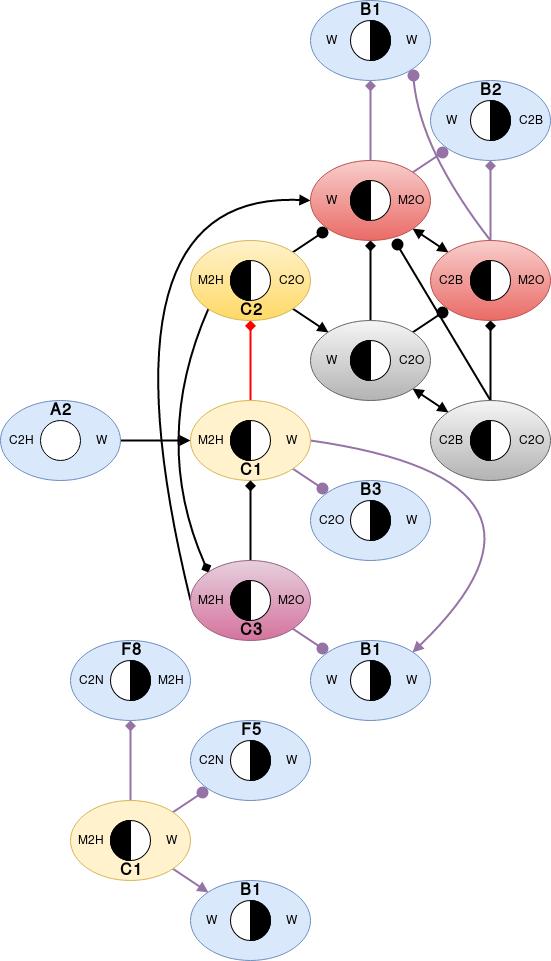}
	\caption{\textbf{FAULTY 1} configurations}
	\label{fig:FAULTY1}
\end{figure}

\begin{figure}[htbp]
	\centering
	\includegraphics[width=0.7\columnwidth]{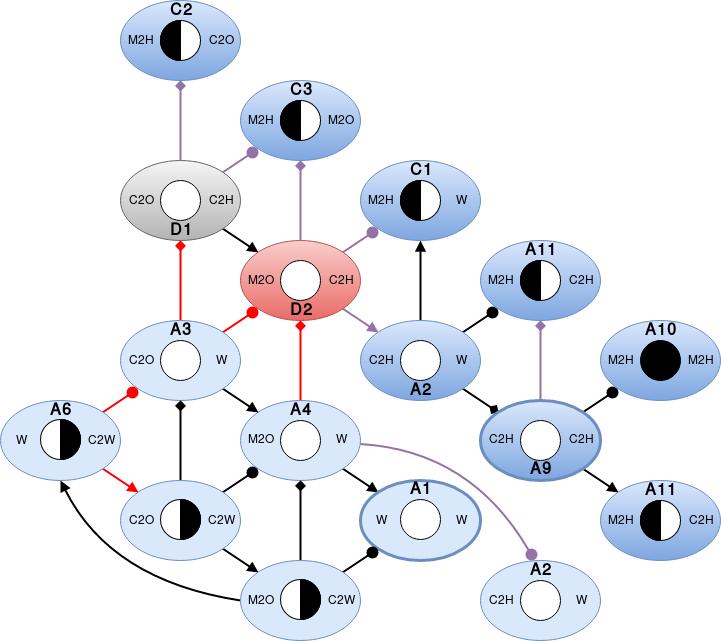}
	\caption{\textbf{FAULTY 2} configurations}
	\label{fig:FAULTY2}
\end{figure}

\begin{figure}[htbp]
	\centering
	\includegraphics[width=0.4\columnwidth]{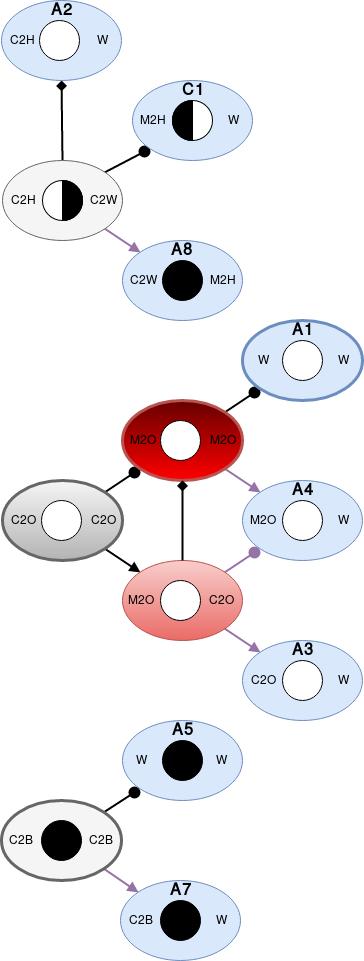}
	\caption{\textbf{ILLEGAL} configurations}
	\label{fig:ILLEGAL}
\end{figure}

\begin{figure}[htbp]
	\centering
	\includegraphics[width=0.7\columnwidth]{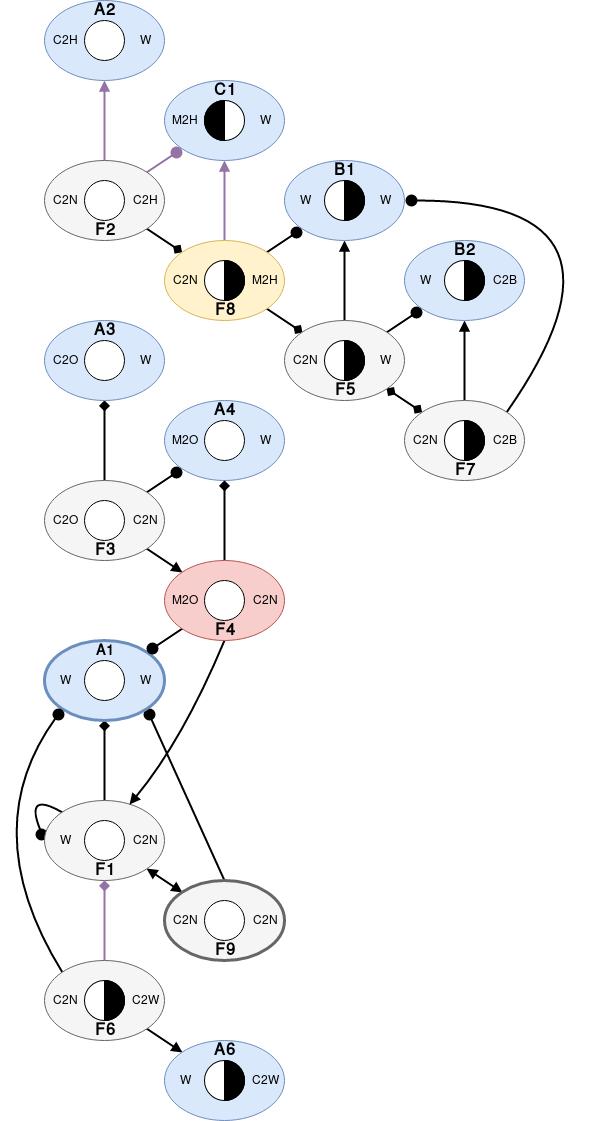}
	\caption{\textbf{GATHERED} configurations}
	\label{fig:GATHERED}
\end{figure}

\subsection{Proof of correctness}

Our main result (Theorem~\ref{thm:gathering}) is directly obtained by Lemmas~\ref{lem:trivial} and \ref{lem:allvalid}, that are presented in the sequel.

\begin{Lemma} %
\label{lem:trivial}
If every subset of configurations is valid, then the Algorithm~\ref{alg1} solves gathering in a self-stabilizing manner.
\end{Lemma}

\begin{proof}
If all subsets are valid and all configurations are included in the subsets, then it means that all configurations eventually lead to gathering. Thus the algorithm solves the gathering problem in a finite number of steps.
\end{proof}

\begin{Lemma}
\label{lem:allvalid}
Every subset of configurations of Algorithm~\ref{alg1} is valid.
\end{Lemma}

\begin{proof}
To prove this, we proceed in order: We first prove the validity of \textbf{ASYM} (Lemma~\ref{lem:ASYM}). We then show that \textbf{FAULTY 1} can only lead to \textbf{ASYM} and is therefore valid (Lemma~\ref{lem:FAULTY1}). 

As \textbf{SYM} and \textbf{FAULTY 2} are interdependent, we use a series of Lemmas to obtain validity results for those subsets (see Corollaries~\ref{cor:FAULTY2} and \ref{cor:SYM} for \textbf{FAULTY 2} and \textbf{SYM} subsets, respectively).

Finally, we prove \textbf{ILLEGAL} (Lemma~\ref{lem:ILLEGAL}) and \textbf{GATHERED} (Lemma~\ref{lem:GATHERED}) to be valid subsets.
\end{proof}

\begin{Lemma}
\label{lem:ASYM}
The \textbf{ASYM} subset is valid.
\end{Lemma}

\begin{proof}
\textbf{ASYM} is the subset reached after a successful de-synchronization between $A$ and $B$, where $A$ and $B$ now have distinct colors with no option to get identical colors again. The \textbf{ASYM} subset has no possible exit to another subset, and from a nominal configuration, fair activations of the robots lead to:

\begin{itemize}
\item $A$ moves towards $B$
\item $B$ does not move
\end{itemize}

As $A$ progresses by at least $\delta$ or reaches $B$, nominal behavior leads to both robots sharing the same coordinates in a finite number of steps. Then, the next activation of $A$ least to either configuration $B1$ or $B2$. One more activation of $A$ leads to either $F5$ or $F7$ configuration, which in turn, leads to $B1$ or $B2$. Overall, the system is stuck in a no-movement infinite loop. So, the \textbf{ASYM} subset is valid.
\end{proof}

\begin{Lemma}
\label{lem:FAULTY1}
The \textbf{FAULTY 1} subset is valid.
\end{Lemma}

\begin{proof}
The \textbf{FAULTY 1} subset is reached after a White to Black de-synchronization. It is then possible to have robot $A$ move towards the midpoint between $A$ and $B$, and $B$ moving towards $A$.

We observe that in the \textbf{FAULTY 1} subset, from a nominal configuration, the system can only exit \textbf{FAULTY 1} through \textbf{ASYM}, which is a valid subset (Lemma~\ref{lem:ASYM}).
It is also worth noticing that, while three cycles exist in \textbf{FAULTY 1}, none of them can actually be traversed forever under the fair scheduler assumption, \emph{i.e.} it would require one robot to never be activated. Also, while it is possible to switch from the first cycle to the second, and from the second to the third, it is not possible to move back to a previous cycle, and breakage of the third cycle leads the system to the \textbf{ASYM} subset.

Therefore, the nominal behavior of \textbf{FAULTY 1} leads to \text{ASYM}, which is valid, and is therefore also valid.

Note that the only configuration of \textbf{FAULTY 1} whose behavior is modified by reaching gathering is $C1$. $C1$ can then lead to either $B1$ or \textbf{ASYM}, or $F5$ and \textbf{ASYM}. Finally, if only the White robot is activated, it is possible to cycle between $C1$ and $F8$. However, the fair scheduler assumption requires this cycle to be broken and the system to move to the \textbf{ASYM} subset.
We can conclude that, despite being a faulty subset, \textbf{FAULTY 1} is merely a temporary subset leading to \textbf{ASYM}, provided that the scheduler is fair.
\\

Another possible proof of the validity of \textbf{FAULTY 1} is found in \cite{OWK17r}, Lemma 2. It is proven that if robot A performs a LOOK right after B performed a COMPUTE while robots are of different colors, gathering is unavoidable. This happens when reaching C2 from either C1 or A2. Since the only paths in \textbf{FAULTY 1} either crosses C2 or immediately exits through \textbf{ASYM}, it is easily proven that \textbf{FAULTY 1} is valid.
\end{proof}

\begin{Lemma}
\label{lem:SYM}
If we exclude the exit to \textbf{FAULTY 2}, the \textbf{SYM} subset is valid.
\end{Lemma}

\begin{proof}
The \textbf{SYM} subset is considered the normal starting starting subset ($A$ and $B$ waiting in White). It includes the FSYNC states ($A1$ to $A1$ through simultaneous activations), and every state that can be reached by $A1$ and lead back to $A1$ (except through \textbf{FAULTY 2} by Hypothesis).  

In \textbf{SYM}, it is possible to exit through \textbf{ASYM}, \textbf{FAULTY 1}, or \textbf{FAULTY 2}. It is also possible to cycle through \textbf{SYM}. This can either be done from $A6$ to $A2$, $A7$ to $A5$ or by cycling back to $A1$. In each case, cycling through \textbf{SYM} implies both robots target the midpoint while not moving, and then both of them moving towards the midpoint. As the distance between robots decreases by $2\delta$ in each cycle, this behavior is valid.
We have already proven \textbf{ASYM} and \textbf{FAULTY 1} to be valid (Lemmas~\ref{lem:ASYM} and \ref{lem:FAULTY1}). Therefore, \textbf{SYM} nominal behavior is valid if we do not consider exits to \textbf{FAULTY 2}.

Achieving gathering implies terminal behavior in configurations $A1$, $A2$, $A3$, $A4$, and $A6$. Possible terminal behaviors are:

\begin{itemize}
\item $A1$ can lead to $F1$, which leads back to $A1$, 
\item $A1$ can lead to $F9$ which leads back to $A1$, 
\item $A2$ can lead to \textbf{ASYM},
\item $A2$ can lead to $F2$ or $F8$, which leads back to $A2$, \textbf{FAULTY 1}, or \textbf{ASYM},
\item $A3$ can lead to $F3$ or $F4$, which leads back to $A3$, $A4$ or $A1$,
\item $A3$ can lead to $A4$, 

\item $A4$ can lead to $A1$,
\item $A4$ can lead to $F1$, which leads to $A1$, 
\item $A4$ can lead to $F4$, which leads to $A1$ or back to $A4$,

\item $A6$ can lead to $A1$, 
\item $A6$ can lead to $F1$, which leads to $A1$, 
\item $A6$ can lead to $F6$, which leads to $A1$ or back to $A6$.
\end{itemize}

From a terminal configuration, it is thus impossible to reach \textbf{FAULTY 2}, as neither $A6$ nor $A2$ can lead to their faulty targets any longer. Furthermore, Reaching $A2$ with an invalid target after achieving gathering now leads to \textbf{FAULTY 1} or \textbf{ASYM}, which are valid subsets. In turn, looping though $A2$ requires an unfair scheduler.
Let us observe that, once gathering is complete, \emph{SYM} can either loop without motion, or lead to a valid subset. Therefore, \textbf{SYM} has a valid terminal behavior.

Overall, \textbf{SYM}, excluding exits to \textbf{FAULTY 2}, is a valid subset.
\end{proof}

\begin{Lemma}
\label{lem:FAULTY2}
If we exclude the exits to \textbf{SYM}, the \textbf{FAULTY 2} subset is valid.
\end{Lemma}

\begin{proof}
Excluding the path to \textbf{SYM} through $A2$ leaves \textbf{FAULTY 1} as the only possible exit for \textbf{FAULTY 2}. Because we have proven \textbf{FAULTY 1} to be valid (Lemma~\ref{lem:FAULTY1}), \textbf{FAULTY 2} is also valid.
\end{proof}

\begin{Lemma}
\label{lem:distancedecreases}
The \textbf{FAULTY 2} $\rightarrow$ \textbf{SYM} $\rightarrow$ \textbf{FAULTY 2} cycle reduces the distance $x$ between $A$ and $B$ by at least $min(\dfrac{x}{2},3\delta)$.
\end{Lemma}

\begin{proof}
Looping through the \textbf{FAULTY 2} $\rightarrow$ \textbf{SYM} $\rightarrow$ \textbf{FAULTY 2} cycle requires to traverse configuration $A6$, and then configuration $A2$.

Let us assume that in $A6$, we have:
\begin{itemize}
\item $A$ is White, in position $0$ with no target,
\item $B$ is Black, in position $x$ with no target.
\end{itemize}

Then, by exploring every possible path, when reaching $A2$, the following holds:

\begin{itemize}
\item $A$ is in position $[min(\delta,x),x]$ with no target,
\item $B$ is in position $x$ with either $\dfrac{x}{2}$ or $[\dfrac{x}{2},x]$ as a target.
\end{itemize}

Now, to go back to $A6$, the system first needs to exit $A2$, which necessarily makes $A$ target the midpoint between $A$ and $B$. It then needs to execute both targets. 
This implies that, when back in $A6$, the distance between $A$ and $B$ is now in the interval $[-\dfrac{x}{2},x-3\delta]$.
\end{proof}

\begin{Lemma}
\label{lem:smalldistance}
After reaching $A6$ with a distance less than $\delta$, reaching $A2$ again implies gathering has been achieved.
\end{Lemma}

\begin{proof}
If the distance $x$ between $A$ and $B$ in configuration $A6$ is less than $\delta$, then at the end of $D2$, gathering is achieved, as $A$ would have moved to $B$. However, a robot still has an invalid target as it enters $A2$.
\end{proof}

\begin{Lemma}
\label{lem:FAULTY2terminal}
Repeating the \textbf{FAULTY 2} $\rightarrow$ \textbf{SYM} $\rightarrow$ \textbf{FAULTY 2} cycle leads to a terminal configuration.
\end{Lemma}

\begin{proof}
As we have proven in Lemma~\ref{lem:distancedecreases}, looping this cycle reduces the distance by at least $min(\dfrac{x}{2},3\delta)$. Because of this, an scheduler repeating this pattern cannot prevent the robots from getting (strictly) closer than $\delta$ from one another in a finite number of steps. Once this has happened, the next use of the cycle necessarily leads to a gathering (by Lemma~\ref{lem:smalldistance}).
The nominal behavior of FAULTY 2 is then valid.
\end{proof}

\begin{Corollary}
\label{cor:FAULTY2}
\textbf{FAULTY 2} is valid.
\end{Corollary}

\begin{proof}
By Lemma~\ref{lem:FAULTY2terminal}, the nominal behavior of \textbf{FAULTY 2} is valid. 

Now, \textbf{FAULTY 2} does not have a terminal behavior. We can, however, note that in terminal behavior, $A2$ cannot lead to $A6$, as we have proven in Lemma~\ref{lem:SYM}. This implies that the cycle leads to gathering, and is then broken when reaching $A2$. 

Therefore, \textbf{FAULTY 2} is a valid subset.
\end{proof}

\begin{Corollary}
\label{cor:SYM}
\textbf{SYM} is valid.
\end{Corollary}

\begin{proof}
From Corollary~\ref{cor:FAULTY2} and Lemma~\ref{lem:SYM}.
\end{proof}

\begin{Lemma}
\label{lem:ILLEGAL}
The \textbf{ILLEGAL} subset is valid.
\end{Lemma}

\begin{proof}
The \textbf{ILLEGAL} subset covers configurations that cannot be reached except through transient errors. Proving its validity is necessary to ensure self-stabilization. 
We can quickly notice that no cycle exists within the \textbf{ILLEGAL} subset, and that this subset necessarily exits to either \textbf{SYM} or \textbf{FAULTY 1}.
As those subsets are valid (Corrolary~\ref{cor:SYM} and Lemma~\ref{lem:FAULTY1}), the \textbf{ILLEGAL} subset is also valid.
\end{proof}

\begin{Lemma}
\label{lem:GATHERED}
The \textbf{GATHERED} subset is valid.
\end{Lemma}

\begin{proof}
This \textbf{GATHERED} subset corresponds to the terminal configuration of all other subsets, which have been proved to be valid. 
\end{proof}

\section{Conclusions}

In this paper we presented a new algorithm to solve gathering of two robots in a self-stabilizing fashion for the non-rigid ASYNC model with an optimal number of lights. Despite being simple to describe, its proof (and the possibility to start from any possible configuration in an asynchronous execution model) mandated significant effort and systematic consideration of all possible cases. 

A natural open question is the possibility of self-stabilizing gathering with an arbitrary number of robots (in particular, a solution that starts from a bivalent configuration with two piles of $\frac{n}{2}$ robots would be of interest), and the minimal number of colors that is required in case the problem is solvable. However, we expect the complexity of the proof to go beyond what is tractable by a human, and would like to consider the possibility of using formal methods. Currently, model-checking~\cite{devismes12sss,berard16dc,doan16sofl,sangnier17fmcad} and program synthesis~\cite{bonnet14wssr,MPST14c} cannot scale to an arbitrary number of robots, and proof assistant techniques~\cite{auger13sss,courtieu15ipl,courtieu16disc,balabonski16sss} do not yet permit to reason about the ASYNC model. Most likely, solving self-stabilizing gathering with $n$ robots in ASYNC will require significant advances in mobile robot formalization.


\begin{thebibliography}{}

\end{thebibliography}


\begin{thebibliography}{10}

\bibitem{AP06j}
Noa Agmon and David Peleg.
\newblock Fault-tolerant gathering algorithms for autonomous mobile robots.
\newblock {\em {SIAM} J. Comput.}, 36(1):56--82, 2006.

\bibitem{auger13sss}
C{\'{e}}dric Auger, Zohir Bouzid, Pierre Courtieu, S{\'{e}}bastien Tixeuil, and
  Xavier Urbain.
\newblock Certified impossibility results for byzantine-tolerant mobile robots.
\newblock In Teruo Higashino, Yoshiaki Katayama, Toshimitsu Masuzawa, Maria
  Potop{-}Butucaru, and Masafumi Yamashita, editors, {\em Stabilization,
  Safety, and Security of Distributed Systems - 15th International Symposium,
  {SSS} 2013, Osaka, Japan, November 13-16, 2013. Proceedings}, volume 8255 of
  {\em Lecture Notes in Computer Science}, pages 178--190. Springer, 2013.

\bibitem{balabonski16sss}
Thibaut Balabonski, Am{\'{e}}lie Delga, Lionel Rieg, S{\'{e}}bastien Tixeuil,
  and Xavier Urbain.
\newblock Synchronous gathering without multiplicity detection: {A} certified
  algorithm.
\newblock In Borzoo Bonakdarpour and Franck Petit, editors, {\em Stabilization,
  Safety, and Security of Distributed Systems - 18th International Symposium,
  {SSS} 2016, Lyon, France, November 7-10, 2016, Proceedings}, volume 10083 of
  {\em Lecture Notes in Computer Science}, pages 7--19, 2016.

\bibitem{berard16dc}
B\'{e}atrice B\'{e}rard, Pascal Lafourcade, Laure Millet, Maria Potop-Butucaru,
  Yann Thierry-Mieg, and S\'{e}bastien Tixeuil.
\newblock Formal verification of mobile robot protocols.
\newblock {\em Distributed Computing}, 2016.

\bibitem{BCM15c}
Subhash Bhagat, Sruti~Gan Chaudhuri, and Krishnendu Mukhopadhyaya.
\newblock Fault-tolerant gathering of asynchronous oblivious mobile robots
  under one-axis agreement.
\newblock In M.~Sohel Rahman and Etsuji Tomita, editors, {\em {WALCOM:}
  Algorithms and Computation - 9th International Workshop, {WALCOM} 2015,
  Dhaka, Bangladesh, February 26-28, 2015. Proceedings}, volume 8973 of {\em
  Lecture Notes in Computer Science}, pages 149--160. Springer, 2015.

\bibitem{bonnet14wssr}
Fran{\c{c}}ois Bonnet, Xavier D{\'{e}}fago, Franck Petit, Maria
  Potop{-}Butucaru, and S{\'{e}}bastien Tixeuil.
\newblock Discovering and assessing fine-grained metrics in robot networks
  protocols.
\newblock In {\em 33rd {IEEE} International Symposium on Reliable Distributed
  Systems Workshops, {SRDS} Workshops 2014, Nara, Japan, October 6-9, 2014},
  pages 50--59. {IEEE}, 2014.

\bibitem{BDT13c}
Zohir Bouzid, Shantanu Das, and S\'{e}bastien Tixeuil.
\newblock Gathering of mobile robots tolerating multiple crash faults.
\newblock In {\em Proceedings of the IEEE International Conference on
  Distributed Computing Systems (ICDCS 2013)}, pages 337--346, Philadelphia,
  PA, USA, July 2013. IEEE Press.

\bibitem{BT15}
Quentin Bramas and S\'ebastien Tixeuil.
\newblock Wait-free gathering without chirality.
\newblock In {\em Structural Information and Communication Complexity - 22nd
  Intl. Colloquium (SIROCCO), Post-Proceedings}, volume 9439, pages 313--327,
  Montserrat, Spain, July 2015.

\bibitem{CRTU15j}
Pierre Courtieu, Lionel Rieg, S\'{e}bastien Tixeuil, and Xavier Urbain.
\newblock Impossibility of gathering, a certification.
\newblock {\em Information Processing Letters (IPL)}, 115(3):447--452, January
  2015.

\bibitem{courtieu16disc}
Pierre Courtieu, Lionel Rieg, S{\'{e}}bastien Tixeuil, and Xavier Urbain.
\newblock Certified universal gathering in {\textbackslash}mathbb {R} {\^{}}2
  for oblivious mobile robots.
\newblock In Cyril Gavoille and David Ilcinkas, editors, {\em Distributed
  Computing - 30th International Symposium, {DISC} 2016, Paris, France,
  September 27-29, 2016. Proceedings}, volume 9888 of {\em Lecture Notes in
  Computer Science}, pages 187--200. Springer, 2016.

\bibitem{courtieu15ipl}
Pierre Courtieu, Lionel Rieg, Sébastien Tixeuil, and Xavier Urbain.
\newblock {Impossibility of Gathering, a Certification}.
\newblock {\em Information Processing Letters}, 115:447--452, 2015.

\bibitem{DAS2016171}
Shantanu Das, Paola Flocchini, Giuseppe Prencipe, Nicola Santoro, and Masafumi
  Yamashita.
\newblock Autonomous mobile robots with lights.
\newblock {\em Theoretical Computer Science}, 609:171 -- 184, 2016.

\bibitem{DGC+15}
Xavier D\'efago, Maria Gradinariu Potop-Butucaru, Julien Cl\'ement, St\'ephane
  Messika, and Philippe Raipin-Parv\'edy.
\newblock Fault and byzantine tolerant self-stabilizing mobile robots
  gathering.
\newblock Research Report IS-RR-2015-003, Japan Adv. Inst. of Science and Tech.
  (JAIST), Hokuriku, Japan, February 2015.

\bibitem{DGM+06}
Xavier D\'efago, Maria~Gradinariu Potop-Butucaru, St\'ephane Messika, and
  Philippe Raipin-Parv\'edy.
\newblock Fault-tolerant and self-stabilizing mobile robots gathering:
  Feasibility study.
\newblock In {\em Proc. 20th Intl. Symp. Distributed Computing (DISC)}, volume
  LNCS 4167, pages 46--60, September 2006.

\bibitem{devismes12sss}
St{\'{e}}phane Devismes, Anissa Lamani, Franck Petit, Pascal Raymond, and
  S{\'{e}}bastien Tixeuil.
\newblock Optimal grid exploration by asynchronous oblivious robots.
\newblock In Andr{\'{e}}a~W. Richa and Christian Scheideler, editors, {\em
  Stabilization, Safety, and Security of Distributed Systems - 14th
  International Symposium, {SSS} 2012, Toronto, Canada, October 1-4, 2012.
  Proceedings}, volume 7596 of {\em Lecture Notes in Computer Science}, pages
  64--76. Springer, 2012.

\bibitem{DP12j}
Yoann Dieudonn{\'{e}} and Franck Petit.
\newblock Self-stabilizing gathering with strong multiplicity detection.
\newblock {\em Theor. Comput. Sci.}, 428:47--57, 2012.

\bibitem{doan16sofl}
Ha~Thi~Thu Doan, Fran{\c{c}}ois Bonnet, and Kazuhiro Ogata.
\newblock Model checking of a mobile robots perpetual exploration algorithm.
\newblock In Shaoying Liu, Zhenhua Duan, Cong Tian, and Fumiko Nagoya, editors,
  {\em Structured Object-Oriented Formal Language and Method - 6th
  International Workshop, {SOFL+MSVL} 2016, Tokyo, Japan, November 15, 2016,
  Revised Selected Papers}, volume 10189 of {\em Lecture Notes in Computer
  Science}, pages 201--219, 2016.

\bibitem{FPS12b}
Paola Flocchini, Giuseppe Prencipe, and Nicola Santoro.
\newblock {\em Distributed Computing by Oblivious Mobile Robots}.
\newblock Synthesis Lectures on Distributed Computing Theory. Morgan {\&}
  Claypool Publishers, 2012.

\bibitem{MPST14c}
Laure Millet, Maria Potop{-}Butucaru, Nathalie Sznajder, and S{\'{e}}bastien
  Tixeuil.
\newblock On the synthesis of mobile robots algorithms: The case of ring
  gathering.
\newblock In Pascal Felber and Vijay~K. Garg, editors, {\em Stabilization,
  Safety, and Security of Distributed Systems - 16th International Symposium,
  {SSS} 2014, Paderborn, Germany, September 28 - October 1, 2014. Proceedings},
  volume 8756 of {\em Lecture Notes in Computer Science}, pages 237--251.
  Springer, 2014.

\bibitem{OWK17r}
Takashi Okumura, Koichi Wada, and Yoshiaki Katayama.
\newblock Optimal asynchronous rendezvous for mobile robots with lights.
\newblock Technical report, 2017.

\bibitem{sangnier17fmcad}
Arnaud Sangnier, Nathalie Sznajder, Maria Potop{-}Butucaru, and S{\'{e}}bastien
  Tixeuil.
\newblock Parameterized verification of algorithms for oblivious robots on a
  ring.
\newblock In {\em Formal Methods in Computer Aided Design}, Vienna, Austria,
  2017.

\bibitem{suzuki99}
Ichiro Suzuki and Masafumi Yamashita.
\newblock Distributed anonymous mobile robots: Formation of geometric patterns.
\newblock {\em SIAM Journal on Computing}, 28(4):1347--1363, 1999.

\bibitem{Viglietta2014}
Giovanni Viglietta.
\newblock Rendezvous of two robots with visible bits.
\newblock In Paola Flocchini, Jie Gao, Evangelos Kranakis, and Friedhelm Meyer
  auf~der Heide, editors, {\em Algorithms for Sensor Systems: 9th International
  Symposium on Algorithms and Experiments for Sensor Systems, Wireless Networks
  and Distributed Robotics, ALGOSENSORS 2013, Sophia Antipolis, France,
  September 5-6, 2013, Revised Selected Papers}, pages 291--306, Berlin,
  Heidelberg, 2014. Springer Berlin Heidelberg.

\end{thebibliography}
\end{document}